\documentclass[11pt,a4paper,english]{article}
\usepackage[T1]{fontenc}
\usepackage[margin=1in]{geometry}
\usepackage{amsmath}
\usepackage{amsthm}
\usepackage{amstext}
\usepackage{amssymb}
\usepackage[unicode=true,pdfusetitle,
 bookmarks=true,bookmarksnumbered=false,bookmarksopen=false,
 breaklinks=false,pdfborder={0 0 1},backref=false,colorlinks=false]
 {hyperref}
 
{\begin{list}{}%
         {\setlength{\leftmargin}{#1}}%
         \item[]%
}
{\end{list}}

\numberwithin{equation}{section}
\numberwithin{figure}{section}

\theoremstyle{plain}
\newtheorem{theorem}{Theorem}[section]
\newtheorem{corollary}[theorem]{Corollary}
\newtheorem{lemma}[theorem]{Lemma}
\newtheorem{proposition}[theorem]{Proposition}
\newtheorem{claim}[theorem]{Claim}

\theoremstyle{remark}
\newtheorem{remark}[theorem]{Remark}

\theoremstyle{definition}
\newtheorem{definition}[theorem]{Definition}

\newcommand{\lab}{\mathrm{lab}}
\newcommand{\girth}{\mathrm{girth}}
\newcommand{\pseudoedge}[1]{\stackrel{#1}{\longrightarrow}}
\newcommand{\floor}[1]{\left\lfloor #1 \right\rfloor}
\newcommand{\ceil}[1]{\left\lceil #1 \right\rceil}
\newcommand{\myexponent}{\frac{1}{\floor{n/4}}}
\newcommand{\E}{\mathbb{E}}

\begin{document}
\setcounter{page}{0}
\thispagestyle{empty}
\title{Set membership with a few bit probes}

\author{ Mohit Garg and Jaikumar Radhakrishnan \\ 
	Tata Institute of Fundamental Research, Mumbai\\
	\texttt{\{garg,jaikumar\}@tifr.res.in}
}
\date{}

\maketitle

\begin{abstract}
We consider the bit-probe complexity of the set membership problem,
where a set $S$ of size at most $n$ from a universe of size $m$ is to
be represented as a short bit vector in order to answer membership
queries of the form ``Is $x$ in $S$?'' by {\em adaptively} probing the
bit vector at $t$ places.  Let $s(m,n,t)$ be the minimum number of
bits of storage needed for such a scheme.  Several recent works
investigate $s(m,n,t)$ for various ranges of the parameter; we obtain
the following improvements over the bounds shown by Buhrman,
Miltersen, Radhakrishnan, and Srinivasan~\cite{BMRV2002} and Alon and
Feige~\cite{AF2009}.
\noindent \paragraph{For two probes ($t=2$):} 
\begin{enumerate}
\item[(a)] $s(m,n,2) =O(m^{1-\frac{1}{4n+1}})$;  this improves on a
  result of Alon and Feige that states that for $n \leq \lg m$,
  $s(m,n,2) = O( m n \lg ((\lg m) / n) / \lg m)$.

\item[(b)] $s(m,n,2)= \Omega(m^{1-\frac{1}{\lfloor n/4
      \rfloor}})$; in particular, $s(m,n,2)=\Omega(m)$ for $n \geq \lg
  m$, that is, if $s(m,n,2) = o(m)$ (significantly better than the
  characteristic vector representation), then $n =
  o(\lg m)$.
\end{enumerate}

\noindent \paragraph{For three probes ($t=3$):} $s(m,n,3) = O(\sqrt{m
	n \lg \frac {2m}{n}}).$ This improves a result of Alon and Feige
that states that $s(m,n,2)=O(m^{\frac{2}{3}}
n^{\frac{1}{3}})$.  

%\marginpar{The improvement is only if $n \ll m$, unless we are able
%to introduce $\lg (em/n)$.}

\noindent\paragraph{In general:} 
%\marginpar{cite previous results for part (a) and (b); state the
%conditions on $n$ etc. for part (c)}
\begin{enumerate}
\item[(a)] (Non-adaptive schemes) For odd $t\geq 5$, there is a
  non-adaptive scheme using $O(t m^{\frac{2}{t-1}}$
  $n^{1-\frac{2}{t-1}} \lg \frac {2m}{n})$ bits of space.  This
  improves on a result of Buhrman et al.~\cite{BMRV2002} that states
  that for odd $t\geq 5$, there exists a non-adaptive scheme that uses
  $O(tm^{\frac {4}{t+1}} n)$ bits of space.
 
\item[(b)] (Adaptive schemes) For odd $t \geq 3$ and $t \leq
  \frac{1}{10} \lg\lg m$ and for $n \leq m^{1-\epsilon}$ ($\epsilon
  > 0$), we have $s(m,n,t)=O(\exp(e^{2t})m^{\frac 2 {t+1}} n^{1 -
    \frac 2 {t+1}} \lg m)$. Previously, for $t \geq 5$, no adaptive
  scheme was known that was more efficient than the non-adaptive
  scheme due to Buhrman et al.~\cite{BMRV2002}, which uses
  $O(tm^\frac4{t+1}n)$ bits of space.

\item[(c)] If $t \geq 3$ and $4^t\leq n$, then $\displaystyle s(m,n,t)
  \geq \frac{1}{15}m^{\frac{1}{t-1}(1-\frac{4^{t}}{n})}.$ For $n \leq
  \lg m$, this improves on the lower bound
  $s(m,n,3)=\Omega(\sqrt{mn/\lg m})$ (valid only for $n \geq 16\lg
  m$ and for {\em non-adaptive} schemes) due to Alon and Feige; for
  small values of $n$, it also improves on the lower bound
  $s(m,n,t)=\Omega(tm^{\frac{1}{t}}n^{1-\frac{1}{t}})$ due to Buhrman
  et al.~\cite{BMRV2002}.
\end{enumerate}
\end{abstract}


\newpage

\section{Introduction}
\setcounter{page}{1}

We study the static set membership problem: given a subset $S$ of
$[m]$ represent it in memory so that membership queries can be
answered using a small number of bit probes (we assume random access
is allowed into the memory). Standard solutions to the set membership
problem can be examined in this light. (We use $\lg$ to mean logarithm
to the base two.)
\begin{description}
\item[The characteristic vector:] Sets can be represented as a
  bit-string of length $m$, and membership queries are answered using
  a single bit probe. However, this representation in not sensitive to
  the number of elements in the set, which can be much smaller than
  $m$.

\item[The sorted table:] Suppose the set $S$ has $n$ elements. Using
  the standard representation of elements of the universe in $\lg m$
  bits, we may store $S$ in memory as a sorted table of $n \lg m$
  bits. Queries can then be answered using binary search taking about
  $(\lg m)(\lg n)$ bit probes in the worst case.
\end{description}
The static membership problem in the bit probe model (in contrast to
the more common cell-probe model) was already studied (in the average
case) by Minsky and Papert in their 1969 book {\em
  Perceptrons}~\cite{MP1969}.  More recently, the worst-case
space-time trade-off for this problem was considered by Buhrman,
Miltersen, Radhakrishnan and Venkatesh~\cite{BMRV2002} and in several
subsequent works~\cite{AF2009, LMNR, RRR2001, RSS2010, V2012}. The
set membership problem for sets where each element is included with
probability $p$ was considered by Makhdoumi, Huang, M\'{e}dard and
Polyanskiy~\cite{MHMP}; they showed, in particular, that no savings
over the characteristic vector can be obtained in this case for
non-adaptive schemes with $t=2$.

To describe the previous results and our contributions formally, we
will use the following definitions.
\begin{definition}
An $(m,n,s)$-{\em storing scheme} is a method for representing a
subset of size at most $n$ of a universe of size $m$ as an $s$-bit
string.  Formally, an $(m,n,s)$-storing scheme is a map $\phi$ from
${{[m]} \choose {\leq n}}$ to $\{0,1\}^s$.  A deterministic
$(m,s,t)$-{\em query scheme} is a family $\{T_u\}_{u\in [m]}$ of $m$
Boolean decision trees of depth at most $t$. Each internal node in a
decision tree is marked with an index between $1$ and $s$, indicating
the address of a bit in an $s$-bit data structure. For each internal
node, there is one outgoing edge labeled ``0'' and one labeled ``1''.
The leaf nodes of every tree are marked `Yes' or `No'. Such a tree
$T_u$ induces a map from $\{0,1\}^s$ to \{Yes, No\}; this map will
also be referred to as $T_u$.  An $(m,n,s)$-storing scheme $\phi$ and
an $(m,s,t)$-query scheme $\{T_u\}_{u \in [m]}$ together form an {\em
$(m,n,s,t)$-scheme} if $\forall S\in {[m] \choose \leq n},\, \forall
u\in [m]: T_u(\phi(S)) =$Yes if and only if $u \in S$.  Let $s(m,n,t)$
be the minimum $s$ such that there is an $(m,n,s,t)$-scheme\footnote{In the literature this function is often written as $s(n,m,t)$; we list the parameters in alphabetical order.}. 

We say that an $(m,n,s,t)$-scheme is {\em systematic} if the value
returned by each of its trees $T_u$ is equal to the last bit it reads
(interpreting $0$ as No/False and $1$ as Yes/True).
\end{definition}

\begin{remark} Note that this definition describes a non-uniform model
  and ignores the important issue of uniformly representing the
  decision trees in the query algorithm. Furthermore, disregarding the
  fact that in practice memory is organized in words, it instead
  focuses attention on the fundamental trade-off between the
  compactness of information representation and the efficiency of
  information extraction in the context of the set membership
  problem. The upper bounds derived in this model are not always 
  realistic (they sometimes rely on probabilistic existence
  arguments); however, lower bounds derived here are generally
  applicable.
\end{remark}

The main focus of Buhrman {\em et al.} was the randomized version of
the above schemes; they showed that membership queries can be answered
correctly with probability $1-\epsilon$ by making just one bit probe
into a representation of size $O(\frac{n}{\epsilon^2} \lg m)$ bits.
They also showed the following lower and upper bounds for
deterministic schemes,: (i) $s(m,n,t) = \Omega(tm^{\frac{1}{t}}
n^{1-\frac{1}{t}})$ valid when $n \leq m^{1-\epsilon}$ (for $\epsilon
> 0$ and $t \ll \lg m$) and (ii) $s(m,n,t) = O(m^{\frac{4}{t+1}}n)$
for odd $t \geq 5$. However, Buhrman {\em et al.} left open the
question of whether a scheme better than the characteristic vector was
possible for $t=2,\,3,\,4$, and $n$ large.  Alon and
Feige~\cite{AF2009}, in their paper, ``On the power of two, three and
four probes,'' addressed this shortcoming. Our contributions are
closely related to theirs.

For two probes, Alon and Feige~\cite{AF2009} show the following.
\begin{theorem}
For $n < \lg m$, 
$\displaystyle s(m,n,2) = O\left( {mn \lg \ceil {\frac{\lg m}{n}}}/{\lg m}\right).$\\
Thus, $s(m,n,2)=o(m)$, whenever $n=o(\lg m)$.
\end{theorem}
\noindent They state:
\begin{quote}
There are still rather substantial gaps between the upper and lower
bounds for the minimum required space in most cases considered here;
it will be nice to get tighter estimates.  In particular, it will be
interesting to decide if there are adaptive $(m,n,s,2)$-schemes with
$s < m$, for $n > \sqrt{m}/2$, and to identify the behavior of the
largest $n=n(m)$ so that there are adaptive $(m,n,s,2)$-schemes with
$s=o(m)$.
\end{quote}

\noindent In this paper, we address this by showing the following.
(We assume $m$ is large; all asymptotic claims made below hold for
large $m$.) 
\begin{theorem}[Result 1]\label{thm:two-probe}
\begin{enumerate}
\item[(a)] 
% Let $n \leq \frac{1}{40}\lg m$.  
There is a constant $C >0$, such that for all large $m$, $s(m,n,2)
\leq C \cdot m^{1-\frac{1}{4n+1}}$.
\item[(b)] Let $4 \leq n$. There is constant $D >0$, such that for all
  large $m$, $s(m,n,2) \geq D m^{1-{\myexponent}}$.
\end{enumerate}
\end{theorem}
For three probes, Alon and Feige~\cite{AF2009} show that
$s(m,n,3)=O(m^{\frac{2}{3}} n^{\frac{1}{3}})$. Their query scheme is
adaptive and based on random graphs. We show the following.
\begin{theorem}[Result 2] \label{thm:three-probe-ub}
	$s(m,n,3) = O(\sqrt{mn\lg \frac {2m}{n}})$.
\end{theorem}
This scheme is adaptive. For small values of $n$, this result comes
close to the lower bound shown below in Theorem~\ref{thm:multi-probe-lb}.
We further generalize this construction for large values of $t$.
\begin{theorem}[Result 3, non-adaptive schemes]
\label{thm:generalnonadaptive}
For odd $t\geq 5$,
  there is a non-adaptive scheme using $O(t m^{\frac{2}{t-1}}
  n^{1-\frac{2}{t-1}} \lg \frac{2m}{n})$ bits of space.
\end{theorem}
This improves on a
result of Buhrman et al.~\cite{BMRV2002} that states that for odd
$t\geq 5$ and $n \leq m^{1-\epsilon}$, there exists a non-adaptive scheme
that uses $O(tm^{\frac {4}{t+1}} n)$ bits of space. These schemes, as
well as the non-adaptive scheme for $t=4$ due to Alon and
Feige~\cite{AF2009}, have implications for the problem studied by
Makhdoumi et al.~\cite{MHMP}; unlike in the case of $t=2$,
siginficant savings are possible if $t\geq 4$, even with non-adaptive
schemes\footnote{We are grateful to Tom Courtade and Ashwin Pananjady
  for this observation.}.
\begin{theorem}[Result 4, adaptive schemes] 
\label{thm:generaladaptive}
For odd $t \geq 3$ and $t \leq
  \frac{1}{10} \lg\lg m$ and for $n \leq m^{1-\epsilon}$ ($\epsilon
  > 0$), we have $s(m,n,t)=O(\exp(e^{2t})m^{\frac 2 {t+1}} n^{1 -
    \frac 2 {t+1}} \lg m)$.
%
% For odd $t \geq 3$ and for $n \leq
%  m^{1-\epsilon}$ ($\epsilon > 0$), we have
%  $s(m,n,t)=O(2^t\exp(e^{2t})m^{\frac 2 {t+1}} n^{1 - \frac 2 {t+1}} \lg
%  m)$.
\end{theorem}
%For odd $t\geq 5$ and $t \leq \frac {1}{10}\lg\lg m$, this improves the
%bound $s(m,n,t) = O(tm^{\frac{4}{t+1}}n)$ of Buhrman {\em et
%  al.}~\cite{BMRV2002}.
%\begin{theorem}[Result 3] \label{thm:multi-probe-ub} For $t\geq 3$, we have
%$s(m,n,t) = O(m^{\frac{2}{t+1}} n^{1-\frac{2}{t+1}}2^{3t}\lg m )$
%for $t \leq \lg \lg \lg m -1$. 
%\end{theorem}
%For $t=4$, this improves on the upper bound $s(m,n,4)=
%O(m^{\frac{2}{3}}n^{\frac{1}{3}})$ of Alon and Feige~\cite{AF2009};

We observe that the two-probe lower bound shown above can be used
to derive slightly better lower bounds for $t \geq 3$. 
\begin{theorem}[Result 5] \label{thm:multi-probe-lb}
If $4^t\leq n$, then $\displaystyle s(m,n,t)
\geq \frac{1}{15}m^{\frac{1}{t-1}(1-\frac{4^{t}}{n})}.$
\end{theorem}
In particular,
for $t=3$ and $n\approx \lg m$, this gives an $\Omega(\sqrt{m})$ bound,
whereas the previous best bound~\cite{BMRV2002} was of the form
$\Omega(tn^{\frac{2}{3}} m^{\frac{1}{3}})$.

\paragraph{What is new, what is old:}
As stated before, this work is closely related to the paper of of Alon
and Feige~\cite{AF2009}. For two probes, they explicitly modeled
their problem using graphs, and translated the high girth of the
graphs to their expansion. This allowed them to use Hall's matching
theorem to avoid conflict while allocating memory locations to
elements of the universe. We borrow the idea of using graphs of
high-girth but we do not reduce the allocation to a matching
theorem. Instead, we observe that the constraints in this case can be
written down as a 2-SAT expression. Furthermore, if the graph has high
girth then this 2-SAT expression must be satisfiable and we will be
able to represent our set successfully. Working with 2-SAT instead of
the matching problem allows us to show a stronger upper bound. For
the lower bound we turn the argument on its head: we show roughly that
any valid two-probe scheme must conceal a certain dense graph that
avoids small cycles. Standard graph theoretic results (the Moore
bound) that relate density and girth then deliver us the lower bound.
We believe this approach via 2-SAT offers a better understanding of
the connection between two-probe schemes and graphs of high girth.

Our three-probe scheme (Theorem~\ref{thm:three-probe-ub}) is based on
the following idea. We must ensure that the data structure returns the
answer `Yes' for all query elements in $S$ and `No' for all elements
not in $R$ (in the end we would want $R=[m]\setminus S$). If $R$ is
small, then this can be arranged using Hall's theorem, by slightly
extending the argument used by Alon and Feige~\cite{AF2009} for their
three-probe scheme. But we still need take care of large $R$. We
notice that the last two probes of a three-probe scheme induce
two-probe schemes (precisely how this comes about is not important
here). We will show that whenever $R$ is large, there is always an
element in it that cannot appear in a short cycle in these two-probe
schemes. That is, we may peel this element away, work on the rest, and
then make appropriate adjustments to accommodate this element.  A form
of this argument has been used in the randomized schemes of Buhrman et
al.~\cite{BMRV2002}; it appears in in the literature in other
contexts, such as Invertible Bloom Lookup Tables~\cite{GM} and graph
based LDPC codes~\cite{LBSS}. Our scheme is not explicit, for it
relies on random graphs that are suitable for the peeling and
matching arguments we employ.

We generalize the above arguments to more than four probes by
considering appropriate random query schemes, and identifying
properties of the resulting random graph that allow us to find the
necessary assignment to correctly represent each possible set.

%\marginpar{Maybe something needs to be said about
%the general upper bounds.}
%The generalization to more than four
%probes involves identification of properties the random graph must
%possess so as to exploit the available non-adaptiveness in the query
%scheme.

\subsection{Other related work}

Some recent work on the bit probe complexity of the set membership
problem has focused on sets of small size. The simplest case for which
tight bounds are not known is $n=2$ and $t=2$: an explicit scheme
showing $s(m,2,2)= O(m^{2/3})$ was obtained by Radhakrishnan, Raman
and Rao~\cite{RRR2001}. Radhakrishnan, Shah and
Shannigrahi~\cite{RSS2010} showed that $s(m,2,2)=
\Omega(m^{4/7})$. They also considered the complexity $s(m,n,t)$ for
$n$ small as $t$ becomes large. These latter results were
significantly improved by Lewenstein, Munro, Nicholson and
Raman~\cite{LMNR}, who, in particular, gave a interesting explicit adaptive
schemes showing that for $t\geq 3$ we have
\[ s(m,2,t) \leq (2^t-1)m^{1/(t-2^{2-t})}.\]
Thus, the exponent of $m$ in their bound for $n=2$ is at most
$(1+\frac{4}{t2^t})\frac{1}{t}$; in contrast, the lower bound of
Theorem~\ref{thm:multi-probe-lb} shows that the exponent is at least
$\frac{1}{t-1} \geq (1+\frac{1}{t})\frac{1}{t}$ when the set size is
much bigger than $4^t\lg m$. 
 Furthermore, for $n\geq 2$, they
obtain explicit schemes showing $s(m,n,t) = O(2^t
m^{1/(t-\min{2\lfloor \lg n\rfloor, n-3/2})})$.

%\subsection{Organization of the paper}
%In the main part of the paper, we present detailed proofs of
%Theorem~\ref{thm:two-probe} (two-probe upper and lower bounds),
%Theorem~\ref{thm:three-probe-ub} (three probe upper bound),
% Theorem~\ref{thm:multi-probe-lb} (general lower bound), 
%Theorem~\ref{thm:generalnonadaptive} (general non-adaptive upper bound) and
%Theorem~\ref{thm:generaladaptive} (general adaptive upper bound). \marginpar{added other results}
%
%\marginpar{Removed Organizationof the paper; it did not look very useful.}

\section{Two-probe upper bound: Proof of
  Theorem~\ref{thm:two-probe}~(a)}
\label{sec:two-probe-ub}

We assume that $n \leq \frac{1}{40}\lg m$, for otherwise, the claim
follows from the trivial bound $s(m,n,2)\leq m$ (taking $C$ large
enough).

Our upper bound is based on dense graphs of high girth. The connection
between graphs of high girth and two-probe schemes was first noticed
by Alon and Feige. They used graphs as templates for their query
schemes, and reduced the existence of a corresponding storing scheme to
the existence of matchings. Exploiting the expansion properties of
small sets in graphs of large girth, they then showed that the
necessary matchings do exist. Our query scheme is essentially the same
as theirs. However, we sharpen their analysis and observe that the
storing problem reduces to a 2-SAT instance. The underlying graph's
high girth this time implies that the 2-SAT instance has the necessary
satisfying assignment.

\begin{definition}[Query graph]
An $(m,s)$-query graph is a graph $G$ with three sets of vertices $A$,
$A_0$ and $A_1$, each with $s$ vertices. Each vertex $v \in A$ has
even degree. With each element $x \in [m]$ we associate a triple
$(i(x),i_0(x),i_1(x)) \in A \times A_0 \times A_1 $ such that
$\{i(x),i_0(x)\}, \{i(x),i_1(x)\} \in E(G)$. We label both these edges with
$x$, and require that no edge receive more than one label.

An $(m,s)$-query graph immediately gives rise to a systematic query
scheme. The scheme uses three arrays $A$, $A_0$
and $A_1$ each containing $s$ bits. The query tree $T_x$ processes the
query ``Is $x$ in $S$?'' as follows: if $A[i(x)]$ then $A_1[i_1(x)]$
else $A_0[i_0(x)]$. We use ${\cal T}_{G}$ to refer to this query scheme.
%has the
%following form: the label of the root corresponds to location $i(x)$
%of $A$, the label of the left child of the root (which is visited when
%$0$ is read on the first probe) corresponds to location $i_0(x)$ of
%$A_0$, and label of the right child corresponds to location $i_1(x)$
%of $A_1$.  
We say that the query scheme ${\cal T}_{G}$ is satisfiable for a set $S
\subseteq [m]$, if there is an assignment to the arrays $A$, $A_0$ and
$A_1$ such that all queries of the form ``Is $x$ in $S$?'' are
answered correctly by ${\cal T}_{G}$.
\end{definition}

\begin{proposition}
If there is a $(m,s)$-query graph such that the query scheme ${\cal
T}_{G}$ is satisfiable for all sets $S \subseteq [m]$ of size at most
$n$, then $s(m,n,2) \leq 3s$.
\end{proposition}

Our claim will thus follow immediately if we establish the following
two lemmas.
\begin{lemma} \label{lm:satisfiablility}
Let $G$ be an $(m,s)$-query graph and $S \subseteq [m]$. If 
$\girth(G)> 4|S|$, then ${\cal T}_{G}$ is satisfiable for $S$.
\end{lemma}

\begin{lemma} \label{lm:highgirthquerygraph}
There is an $(m, O(m^{1+\frac{1}{4n+1}}))$-query graph with girth more
than $4n$.
\end{lemma}

\begin{proof}[Proof of Lemma~\ref{lm:satisfiablility}]
%If $S$ is empty, $\psi_S$ is satisfied by the all zeroes assignment.
Fix a non-empty set $S$ of size at most $n$. We need to assign values
to the bits of $A$, $A_0$ and $A_1$ so that all queries are answered
correctly. Note that since our query scheme is systematic, the only
constraints we have are the following.
\begin{description}
\item[$x \in S$:] 
\begin{eqnarray}
\neg A[i(x)] &\rightarrow& A_0[i_0(x)] ;\label{eq:x0} \\
A[i(x)] &\rightarrow& A_1[i_1(x)]. \label{eq:x1}
\end{eqnarray}
\item[$y \not \in S$:] 
\begin{eqnarray}
\neg A[i(y)] &\rightarrow& \neg A_0[i_0(y)]; \label{eq:y0} \\
A[i(y)] &\rightarrow& \neg A_1[i_1(y)].  \label{eq:y1} 
\end{eqnarray}
\end{description}
Let us examine the implications of the above constraints for the
variables from the first array: $A[1], A[2], \ldots, A[s]$.  From
(\ref{eq:x0}) and (\ref{eq:y0}), we conclude that whenever $x \in S$
and $y \not\in S$ and an edge with label $x$ and an edge with label
$y$ meet in $A_0$, we have the constraint
\begin{equation}
 A[i(x)] \vee A[i(y)]. \label{eq:meetinA0}
\end{equation}
Similarly, from (\ref{eq:x1}) and (\ref{eq:y1}), if $x\in S$ and $y
\not \in S$, and an edge with label $x$ and edge with label $y$ meet
in $A_1$, we have the constraint
\begin{equation}
\neg A[i(x)] \vee \neg A[i(y)]. \label{eq:meetinA1}
\end{equation}
Let $\psi_S(A)$ be the 2-SAT instance on variables $A[1], \ldots, A[s]$
consisting of all clauses of the form (\ref{eq:meetinA0}) and
(\ref{eq:meetinA1}). It can be verified that a satisfying assignment
for $\psi_S(A)$ can be extended to the other arrays, $A_0$ and $A_1$,
in order to satisfy all constraints in (\ref{eq:x0})--(\ref{eq:y1}).
So, it suffices to show that $\psi_S(A)$ is satisfiable.

Each clause of the form $x \vee y$ is equivalent to the $\neg x
\rightarrow y$ and $\neg y \rightarrow x$. Furthermore, if $\psi_S(A)$
is not satisfiable, then there must be a chain of such implications
from a literal to its negation (see, e.g., Aspvall, Plass and
Tarjan~\cite{APT1979}).  We now observe that since our graph has large
girth, such a chain cannot exist. Suppose the shortest such chain has
the form
\[ A[i_0] \rightarrow \neg A[i_1] \rightarrow A[i_2] \rightarrow \cdots
\rightarrow A[i_{\ell-1}] \rightarrow \neg A[i_{\ell}],\] where
$i_\ell = i_0$ and otherwise the $i_j$'s are distinct (if they were
not distinct, there would be a shorter chain).  Since each clause of
$\psi_S$ involves at least one element from $S$, we have $\ell \leq
2|S|$.  The first implication corresponds to a path of length two in
$G$ from $A[i_0]$ to $A[i_1]$ via an intermediate vertex in $A_1$, the
second to a path of length two in $G$ from $A[i_1]$ to $A[i_2]$ via
$A_0$, and so on; the last implication corresponds to a path of length
two from $A[i_{\ell-1}]$ to $A[i_\ell]$ via $A_1$.  If $\ell=0$, we
have a path in $G$ from $A[i_0]$ to itself via $A_1$ (consisting of
two different edges, one with label in $S$ and the other with label
not in $S$), resulting in a cycle of length two---a contradiction.  If
$\ell \geq 1$, the first implication shows that there is a path of
length two in $G$ from $A[i_0]$ to $A[i_1]$ via $A_1$.  The remaining
implications show that there is a walk of length $2(\ell-1)$ from
$A[i_1]$ to $A[i_0]$ that starts with an edge from $A[i_1]$ to
$A_0$. Thus, $A[i_1]$ is in a cycle in $G$ of length at most $2\ell
\leq 4|S|$---a contradiction.

A similar argument shows that the shortest such chain cannot be of the
form $\neg A[i_0] \rightarrow  A[i_1] \rightarrow \neg A[i_2] \rightarrow \cdots
\rightarrow \neg A_[i_{\ell-1}] \rightarrow A[i_{0}]$. 
\end{proof}

\begin{proof}[Proof of Lemma~\ref{lm:highgirthquerygraph}]
Let $G=(V_1,V_2,E)$ be a bipartite graph of girth $g$, with $|V_1|,
|V_2| = s$, and each vertex in $V_1$ of even degree.  (Later, we will
indicate how such graphs $G$ can be obtained.)  Let
\begin{eqnarray*}
V_1 &=& \{V_1[1],V_1[2], \ldots, V_1[s]\};\\
V_2 &=& \{V_2[1],V_2[2], \ldots, V_2[s]\}.
\end{eqnarray*}
Consider the $(|E|/2,s)$-query graph $H$ constructed as follows.  $H$
has three vertex sets $A$, $A_0$ and $A_1$. $A$ will be a copy of
$V_1$, and $A_0$ and $A_1$ will be copies of $V_2$. Half the edges of
$G$ between $V_1$ and $V_2$ will be placed between $A$ and $A_0$ and
the rest between $A$ and $A_1$. More precisely, suppose the neighbors
of $V_1[i]$ are $V_2[j_1], V_2[j_2], \ldots, V_2[j_{d}]$. Then, for
$k=1,2,\ldots,d/2$, we include edges $\{A[i],A_0[j_{2k}]\}$ and
$\{A[i],A_1[j_{2k-1}]\}$ in $H$; furthermore, these two edges will
have the same label $x \in [m]$. It is immediate that $H$ is a
$(|E(G)|/2,s)$-query graph with girth at least $g$. Thus, it is enough
to exhibit a bipartite graph $G$ with $|V_1|=|V_2|=
O(m^{1-\frac{1}{4n+1}})$, $|E(G)|=2m$, $\girth(G) > 4n$ and all
vertices in $V_1$ of even degree. We present a probabilistic argument
(essentially due to Erd\"{o}s) to establish the existence of such
graphs.
\paragraph{Dense graphs of large girth:} 
A probabilistic argument (due to Erd\"{o}s) establishes the existence
of such graphs. Let $k=4n \leq \frac{1}{10}\lg m$, and consider the
following random bipartite graph $G$ on vertex sets $V_1$ and $V_2$,
each with $s=\ceil{4m^{1-\frac{1}{k+1}}}$ vertices each. Let $d$ be
the largest even number at most $s^\frac{1}{k}$; thus
$s^{\frac{1}{k}} \geq d > s^{\frac{1}{k}} -2\geq 2$. For each vertex $v \in
  V_1$, we assign $d$ distinct neighbors from $V_2$. Then, the
  expected number of short cycles in $G$ is at most
\[ \sum_{\ell=2}^{k/2} 
\left(\frac{s^{2\ell}}{2\ell}\right) \left(\frac{d}{s}\right)^{2\ell}
\leq \frac{1}{4} \sum_{\ell=2}^{k/2} d^{2\ell} \leq \frac{d^{k}}{4}
\sum_{\ell=0}^{k/2-2} \frac{1}{d^{\ell}} \leq \frac{d^k}{2} \leq
\frac{s}{2}.\] Thus, there is such a graph with at most $\frac{s}{2}$
short cycles. Consider each such cycle one by one, and for each pick
one of its vertices in $V_1$ and delete both edges from the cycle
incident on it.  Then, the average degree in $V_1$ is at least
$d-1 > s^{1/k}-3 \geq \frac{1}{2}s^{1/k}$, and $|E(G)| \geq
\frac{1}{2}s^{\frac{k+1}{k}} \geq 2m$. Using this in the above
construction we obtain $s(m,n,2) \leq \ceil{4m^{1-\frac{1}{4n+1}}}$.
\end{proof} 

\paragraph{Explicit schemes:} Two-probe schemes 
 can be obtained more explicitly using the following construction of
 graphs of large girth.
 \begin{proposition}[see Proposition 2.1 of F. LAZEBNIK, V. A. USTIMENKO, AND A. J. WOLDAR~\cite{LUW}]
Let $q$ be a prime power and $k \geq 1$ be an odd integer. Then, there
is graph $D(k,q)$ that is
\begin{enumerate}
\item[(i)] $q$-regular and of order $2q^k$.
\item[(ii)] $D(k,q)$ has girth at least $k+5$
\end{enumerate}
\end{proposition}
Now assume $ n \leq \frac{1}{4} (\log m)^{1/3}$.
Set $k=4n-3$, and choose $q=2^r$ such that $(2m)^{1/(k+1)} \leq q < 2
(2m)^{1/(k+1)}$. Then, $D(k,q)$ is a bipartite graph on
vertex sets $(V_1,V_2)$, with girth at least $4n+2$, at least $2m$ edges and
\[|V_1|= |V_2| = s = q^k \leq (2(2m))^{1/(k+1)})^k \leq 2^{k+1} m^{1 - 1/(k+1)}.\]
If $n \leq \frac{1}{4} (\log m)^{1/3}$, we have $(k+1)^2(k+2) \leq
  \log m$, and we have
\[ s \leq m^{1/((k+1)(k+2))} m^{1 - 1/(k+1)}  = m^{1 - 1/(k+2)} \leq m^{1 - 1/(4n+1)}.\]

%\begin{remark}
%As stated, the upper bound is not explicit. Somewhat weaker bounds can
%be based on explicit constructions of graphs of high girth
%in~\cite{LUW}.
%\end{remark}

\section{Two-probe lower bound: Proof of
  Theorem~\ref{thm:two-probe}~(b)}
\label{sec:two-probe-lb}

Since $s(m,n,2)$ is an non-decreasing function of $n$, it is enough to
establish the claim for $n \leq \lg m$.

\begin{proposition} \label{prop:systematic}
If there is an $(m,n,s,t)$-scheme, then there is a systematic
$(m,n,2s,t)$-scheme.
\end{proposition}
So, from now on, we will assume that our schemes are systematic.

\begin{definition} (Bipartite graph $H_{\Phi}$,pseudo-graph  $G_{\Phi}$)
Fix a systematic $(m,n,s,2)$-scheme $\Phi$. We will associate the
following bipartite graph $H_\Phi$ with such a scheme. There will be
two sets of vertices, each with $s$ elements: $A_0$ and $A_1$. Each
edge of $H_\Phi$ will have a color and a label. We include the edge
$\{A_0[j],A_1[k]\}$ with label $x$ and color $i$, if on query ``Is $x$
in $S$?'' the first probe is made to location $i$, and if it returns
$0$, the second probe is made to location $j$ and if the first probe
returns a $1$, the second probe is made to location $k$. $H_{\Phi}$
thus has $2s$ vertices and $m$ edges, which are colored using
$s$ colors.

The pseudo-graph $G_{\Phi}$ is a bipartite graph obtained from
$H_{\Phi}$ as follows. $G_{\Phi}$ and $H_{\Phi}$ have the same set of
vertices. The edges of $G_{\Phi}$ are obtained as follows.  Consider
the edges of color $\alpha$ in $H_\Phi$. We partition these edges into
ordered pairs (excluding one edge if the number of edges of this color
is odd). For each such pair we include a pseudo-edge in $G_{\Phi}$ as
follows. Let $(e,e')$ be one such pair; suppose $e=\{u,v\}$ has label
$x$, $e'=\{u',v'\}$ has label $x'$, $u,u' \in A_0$ and $v,v' \in
A_1$. Then, in $G_{\Phi}$ we include the edge $\{u,v'\}$ with label
$\{(u,x), (v',x')\}$ (we {\em do not} include the edge $\{u',v\}$). We
repeat this for all colors $\alpha$.  Thus $G_{\Phi}$ is a bipartite
graph with at least $(m-s)/2$ edges. For a set of edges $P$ of
$G_{\Phi}$, let $\lab(P)\subseteq [m]$ be the set of elements of the
universe that appear in the label of some edge in $P$.
\end{definition}

\subsection{Forcing}

\begin{lemma}[Forcing lemma]\label{lm:forcing}
Let $G_{\Phi}$ be a pseudo-graph associated with a systematic scheme
$\Phi$.  Let $C$ be a cycle in $G_{\Phi}$ starting at vertex $v$. Let
$b \in \{0,1\}$.  Then there are disjoint subsets $S_0, S_1 \subseteq
\lab(C)$, each with at most $|C|+1$ elements, such that in any
representation under $\Phi$ of a set $S$ such that $S_1 \subseteq S
\subseteq \bar{S_0}$, location $v$ must be assigned $b$. Further, if
$b=0$ then $|S_1|=|C|-1$, and if $b=1$, then $|S_1|=|C|+1$.
\end{lemma}
\begin{proof}
First, consider the claim with $b=1$.  Suppose the cycle is
\[v=v_0 \pseudoedge{e_1} v_1 \pseudoedge{e_2} \cdots \pseudoedge{e_{k-1}} v_{k-1}
\pseudoedge{e_k} v_k=v_0,\] and the pseudo-edge $e_i$ has label
$\{(v_{i-1},a_i), (v_i,b_i)\}$.  Let $S_1=\{a_1,a_2,\ldots,a_{k-1},
a_k, b_k\}$ (it has $k+1$ elements) and $S_0 = \{b_1,b_2,\ldots,
b_{k-1}\}$. We claim that if the scheme $\Phi$ represents a set $S$
such that
\[ S_1 \subseteq S \subseteq \bar{S_0},\]
then location $v$ must be assigned $1$.  Suppose location $v$ is
assigned $0$.  Recall that the scheme is systematic. Since $a_1 \in S$
and $b_1 \not\in S$, we conclude from the definition of the
pseudo-edge $e_1$ that $v_1$ must also be assigned $0$.  Using the
subsequent edges in the cycle, we conclude that the locations
$v_0,\ldots,v_{k-1}$ must all be assigned $0$.  Now, however, $a_k,b_k
\in S$, so $1$ must be assigned to location $v_k=v$---a contradiction,
for we assumed that $v$ was assigned $0$. This proves our claim for
$b=1$. For $b=0$, we take $S_1=\{b_1,b_2,\ldots,b_{k-1}\}$ (it has
$k-1$ elements) and $S_0 = \{a_1,a_2,\ldots, a_{k-1},a_k,b_k\}$, and
reason as before.
\end{proof}

\begin{corollary}\label{cor:forcing}
If a scheme stores sets of size up to $2k$, then the pseudo-graph
associated with the scheme cannot have two edge-disjoint cycles of
length at most $k$ each that have a vertex in common.
\end{corollary}
\begin{proof}
Suppose there are two such edge-disjoint cycles, $C_1$ and $C_2$, both
starting at $v$.  We apply Lemma~\ref{lm:forcing} with $b=0$ and
obtain sets $S_0,S_1 \subseteq \lab(C_1)$. Next we apply
Lemma~\ref{lm:forcing} with $b=1$ and obtain sets $T_0,T_1 \subseteq
\lab(C_2)$. Now, consider $S= S_1\cup T_1$, a set of size at most
$2k$. When the scheme stores the set $S$, then 
location $v$ must be assigned a $0$ (because $S_1 \subseteq S \subseteq
\bar{S_0}$), and also $1$ (because $T_1 \subseteq S \subseteq
\bar{T_0}$)---a contradiction.
\end{proof}

\subsection{Calculation}

Our lower bound will use Corollary~\ref{cor:forcing} as follows.  We
will show that if a scheme uses small space to represent sets of size
up to $n$ from a universe of size $m$, then its pseudo-graph must be
dense (for it must accommodate about $m/2$ edges). In such a dense
graph there must be short cycles. In fact, if $m \gg s^{1-4/n}$, then
we can ensure that there are two cycles of length at most $n/2$ each
that have a vertex in common. But, then Corollary~\ref{cor:forcing}
states that such a scheme does not exist.

To make the above argument precise, we will use the following
proposition, which is a consequence of a theorem of Alon, Hoory and
Linial~\cite{AHL} (see also Ajesh Babu and Radhakrishnan~\cite{BR}).
\begin{proposition}
Let $G$ be a bipartite graph average degree $d \geq 2$, and girth
greater than $\floor{\frac{n}{2}}$ (for positive integer $n \geq 4$).
Then, $d \leq (|V(G)|/2)^{{\myexponent}} +1$.
\end{proposition}
\begin{proof}
Let $n=4p+q$, for $p=\floor{\frac{n}{4}}$ and $q$ such that $q\in
\{0,1,2,3\}$. From our assumption, the girth of $G$ is greater than
$\floor{\frac{n}{2}} \geq 2p$. Since $G$ is bipartite, $G$ has girth
at least $2(p+1)$. The result of Alon, Hoory and Linial then
immediately implies that $|V(G)| \geq 2 (d-1)^p$. Since
$p=\floor{n/4}$, our claim follows from this.
\end{proof}
\begin{corollary}\label{cor:two-cycles}
Suppose $G$ is a bipartite graph with $|V(G)|\geq \floor{\frac{n}{2}}$
and $|E(G)| > (|V(G)|/2)^{1+ {\myexponent}}+\frac{3}{2}|V(G)|$ and
$(|V(G)|/2)^{\myexponent} \geq 2$. Then, $G$ has two edge disjoint
cycles of length at most $\frac{n}{2}$ that have at least one vertex
in common.
\end{corollary}
\begin{proof}
If $|E(G)|> (|V(G)|/2)^{1+ {\myexponent}} + |V(G)|/2 $, then the
average degree is at least $(|V(G)|/2)^{\myexponent} + 1 \geq 2$. By
the proposition above, $G$ has a cycle of length $\ell_1 \leq
\floor{\frac{n}{2}}$. Remove this cycle, and consider the remaining
graph, which has more than $(2s)^{1+ {\myexponent}}+ (3|V(G)|/2
-\ell_1)$ edges. Again, we find a cycle of length at most $\ell_2 \leq
\floor{\frac{n}{2}}$. We may continue in this way, finding cycles of
length $\ell_i$ ($i=1,2,\ldots$), until the sum of the $\ell_i$'s
\emph{exceeds} $|V(G)|$. At that point, two of the cycles we found must
intersect (for there are only $|V(G)|$ vertices).
\end{proof}

\begin{proof}{of Theorem~\ref{thm:two-probe} (b)}
Fix an $(m,n,s,2)$-scheme. Assume $m$ is large.  Using
Proposition~\ref{prop:systematic}, we obtain a systematic
$(m,n,2s,2)$-scheme, say $\Phi$, and consider the corresponding
pseudo-graph $G_{\Phi}$ (note $|V(G_{\Phi})|=4s$ and $|E(G)| \geq
(m-s)/2$).  The lower bound of Buhrman et {al.}~\cite{BMRV2002}
implies that $s \geq \sqrt{m}$; thus $(2s)^{\myexponent}\geq
m^\frac{2}{n} \geq 2$; also since we assume $n \leq \lg m$, we have
$4s \geq \floor{\frac{n}{2}}$.  By applying
Corollary~\ref{cor:two-cycles} and Corollary~\ref{cor:forcing} to
$G_{\Phi}$, we conclude that
\[ \frac{m-s}{2} \leq |E(G_\Phi)| \leq (2s)^{1+ {\myexponent}} + 6s.\]
Thus (recall from above that $(2s)^{\myexponent}\geq 2$),
\[m \leq 13s + 2(2s)^{1+{\myexponent}} \leq (7s)^{1+ \myexponent} +
(4s)^{1+ \myexponent} \leq (11s)^{1+ {\myexponent}}.\] 
By raising both
sides to the power $1-\myexponent$ and rearranging the inequality, we
obtain $s \geq \frac{1}{11} m^{1 - {\myexponent}}.$
\end{proof}

\section{Three-probe upper bound: Proof of Theorem~\ref{thm:three-probe-ub}}

As in the case of two probes, our three-probe scheme will be based on
the existence of certain graphs. The framework we use will be general
and also applicable to schemes that make $t \geq 4$ probes. We will
present the general framework first, and specialize it to $t=3$ when
we describe our proof.
\begin{definition} 
An $(m,s,t)$-graph is a bipartite graph $G$ with vertex sets $U=[m]$
and $V$ ($|V|=(2^t-1)s$). $V$ is partitioned into $2^t-1$ disjoint
sets: $A$, $A_0$, $A_1$, $A_{00}$,\ldots, one $A_{\sigma}$ for each
$\sigma \in \{0,1\}^{\leq (t-1)}$; each $A_{\sigma}$ has $s$
vertices. Between each $u \in U$ and each $A_{\sigma}$ there is
exactly one edge. For $i=1,\ldots,t$, the subgraph of $G$ induced by
$U$ and $V_i=\cup_{\sigma: |\sigma|=i-1} A_\sigma$ will be referred to
as $G_i$. An $(m,s,t)$-graph naturally gives rise to a systematic
$(m,(2^t-1)s)$-query scheme ${\cal T}_G$ as follows. We view the
memory (an array ${\mathsf L}$ of $(2^t-1)s$ bits) as being indexed by
vertices in $V$. For query element $u\in U$, if the first ${i-1}$
probes resulted in values $\sigma \in \{0,1\}^{i-1}$, then the $i$-th
probe is made to the location indexed by the unique neighbor of $u$ in
$A_{\sigma}$. In particular, the $i$-th probe is made at a location in
$V_i$.

We say that the query scheme ${\cal T}_{G}$ is satisfiable for a set
$S \subseteq [m]$, if there is an assignment to the memory locations
$({\mathsf L}[v]: v \in V)$, such that ${\cal T}_{G}$ correctly answers all
queries of the form ``Is $x$ in $S$?''.
\end{definition}
We now restrict attention to $t=3$ probes.  First, we identify an
appropriate property of the underlying $(m,s,3)$-graph $G$ that
guarantees that the ${\cal T}_{G}$ is satisfiable for all sets $S$ of
size at most $n$. We then show that such a graph does exist for some
$s = O(\sqrt{mn\lg \frac {2m}{n}})$.
\begin{definition}[Admissible graph] \label{def:admissible}
 We say that an $(m,s,3)$-graph $G$ is admissible for sets of size at
 most $n$, 
if 
\begin{enumerate}
\item[(P1)] $\forall R \subseteq [m]$ ($|R|\leq n + \ceil{n \lg \frac {2m}{n}}$):
  $|\Gamma_{G}(R)| \geq 5|R|$, where $\Gamma_G(R)$ is the set of
  neighbors of $R$ in $G$.

\item[(P2)] 
$\forall S \subseteq [m]\,  (|S| \leq n), \forall R \subseteq
[m]\setminus S\, (|R| > \ceil{n\lg \frac {2m}{n}})\,  \exists y \in R$:
\begin{eqnarray*}
&& \left(\Gamma_{G_3}(y) \cap \Gamma_{G_3}(S) = \emptyset\right) \quad \mbox{OR}
   \\
&& \left(
| \Gamma_{G_3}(y) \cap \Gamma_{G_3}(S)|=1 \quad \mbox{AND} \quad
|\Gamma_{G_1 \cup G_2}(y) \cap \Gamma_{G_1 \cup G_2}( (R\cup S) \setminus \{y\})| \leq 1\right). 
\end{eqnarray*}
\end{enumerate}
\end{definition}

\begin{lemma} \label{lm:admissibleImpliesSatisfiable}
If an $(m,s,3)$-graph $G$ is admissible for sets of size at most $n$,
then the $(m,7s,3)$-query scheme ${\cal T}_{G}$ is satisfiable for
$(S,[m]\setminus S)$ for every $S$ of size at most $n$.
\end{lemma}

\begin{lemma} \label{lm:admissibleExists}
There is an $(m,s,3)$-graph with $s= \ceil{500\sqrt{mn\lg \frac {2m}{n}}}$ that is
admissible for every set $S\subseteq [m]$ of size at most $n$.
\end{lemma}

\begin{proof}[Proof of Lemma~\ref{lm:admissibleImpliesSatisfiable}]
Fix an $(m,s,3)$-graph $G$ that is admissible for sets of size at most $n$.
Thus, $G$ satisfies (P1) and (P2)
above.  Fix a set $S$ of size at most $n$. Suppose ${\cal T}_G$ is
not satisfiable for $S$. Then, there is a minimal set $T \subseteq
[m]\setminus S$ such that ${\cal T}_G$ fails to correctly answer
queries for all $u \in S \cup T$ under every assignment. We have two
cases. 
\begin{description}
\item[$|S \cup T| \leq n+ \ceil{n\lg \frac {2m}{n}}$:] We use an idea from Alon and
  Feige~\cite{AF2009}. From (P1) and Hall's theorem, we may assign
  to each element $u \in S \cup T$ a set $V_u \subseteq \Gamma_G(u)$ such that (i) $|V_u|=5$
  and (ii) the $V_u$'s are disjoint.  It can be verified that in a
  binary decision tree of depth $3$ and any value $b \in \{0,1\}$,
  given any set of FIVE nodes, values can be assigned to those nodes
  to ensure that the tree returns the value $b$. Thus, there is an
  assignment (fixing five bits for each $u \in S \cup T$) so that
  ${\cal T}_G$ returns the correct answer for all $u \in S \cup
  T$---contradicting our choice of $T$.

\item[$|S\cup T| > n+\ceil{(n\lg \frac {2m}{n})}$:] Thus, $T>\ceil{n\lg \frac {2m}{n}}$.
From property (P2), we conclude that there is a $y \in T$ such that
one of the following holds.
\begin{itemize}
\item[(a)] $ \Gamma_{G_3}(y) \cap \Gamma_{G_3}(S) = \emptyset $ or
\item[(b)] $| \Gamma_{G_3}(y) \cap \Gamma_{G_3}(S)|=1 \quad \mbox{AND} \quad  |\Gamma_{G_1 \cup G_2}(y) \cap \Gamma_{G_1 \cup G_2}( (T\cup S) \setminus \{y\})| \leq 1.$
\end{itemize}
By the minimality of $T$, there is an assignment $\sigma \in
\{0,1\}^{V}$ so that ${\cal T}_G$ correctly answers queries for all
elements $u \in S \cup T\setminus \{y\}$. In case (a), modify $\sigma$ so
that all locations in $\Gamma_{G_3}(y)$ (the locations that are probed
in the third step by $T_y$) are $0$. For the new assignment $\sigma'$,
the query for $y$ is clearly answered correctly; the operation of
$T_u$ for $u \in S \cup T\setminus \{y\}$ is identical in $\sigma$ and
$\sigma'$. This again contradicts the choice of $T$.

In case (b), we again start with an assignment $\sigma \in
\{0,1\}^{V}$ so that ${\cal T}_G$ correctly answers queries for all
elements $u \in S \cup T\setminus \{y\}$. Now, to accommodate $y$, we will
modify $\sigma$ to $\sigma'$, by making changes to locations in $V_1$,
$V_2$ and $V_3$. We have exactly one $\ell \in V_3$ such that $\ell
\in \Gamma_{G_3}(y) \cap \Gamma_{G_3}(S)$; in $\sigma'$ all locations
in $\Gamma_{G_3}(y)$ other than $\ell$ are set to $0$. Furthermore, at
least two of the three locations in $\Gamma_{G_1\cup G_2}(y)$ are
outside $\Gamma_{G_1 \cup G_2}( (S\cup T) \setminus \{y\})$; so we may
modify them without affecting the operation of any decision tree
$T_{u}$ for $u \in (S\cup T) \setminus \{y\}$. In $\sigma'$, we assign
these values appropriately so that the third probe of $T_{y}$ is not
$\ell$. 
%(That this can be done is easy to ascertain by inspection; a
%general claim valid for all $t \geq 3$ is verified in
%Section~\ref{sec:multi-probe-ub}.)  
We have thus ensured that under
assignment $\sigma'$, queries for all $u \in S \cup T$ are answered
correctly---again contradicting the choice of $T$.
\end{description}
\end{proof}
\begin{proof}[Proof of Lemma~\ref{lm:admissibleExists}]
We show that a suitable random $(m,s,3)$-graph $G$ is admissible with
positive probability, for $s = \ceil{500 \sqrt{mn\lg \frac {2m}{n}}}$.
The graph $G$ is constructed as follows. Recall that $V =
\cup_{z \in \{0,1\}^{\leq 3}} A_z$. For each $u \in U$, one neighbor is chosen
uniformly and independently from each $A_z$.
\begin{description}
\item[(P1) holds.]  If (P1) fails, then for some non-empty $W
  \subseteq U$, $(|W| \leq n + \ceil{n \lg \frac {2m}{n}})$, we have
  $|\Gamma_G(W)| \leq 5|W| -1$. Fix a set $W$ of size $r \geq 1$ and
  $L \subseteq V$ of size at most $5r-1$. Let $L$ have $\ell_z$
  elements in $A_z$. Then,
\[
\Pr[\Gamma_G(W) \subseteq L] \leq  \prod_z \left(\frac{\ell_z}{|A_z|}\right)^{r} \leq 
\left(\frac{5r-1}{7s}\right)^{7r}.
\]
If $s \geq 500\sqrt{mn\lg \frac {2m}{n}}$, then we conclude, using the union bound
over choices of $W$ and $L$, that the probability that (P1) fails is
at most
\begin{eqnarray*}
&& \sum_{r=1}^{n+\ceil{n\lg \frac {2m}{n}}} {m\choose r} 
{{7s} \choose 5r-1}\left(\frac{5r-1}{7s}\right)^{7r} \\
&\leq& 
  \sum_{r=1}^{n+\ceil{n\lg \frac {2m}{n}}}
  \left(\frac{em}{r}\right)^r\left(\frac{7es}{5r-1}\right)^{5r-1}
  \left(\frac{5r-1}{7s}\right)^{7r} \\
&\leq& 
\sum_{r=1}^{n+\ceil{n\lg \frac {2m}{n}}}
\left(\frac{5r}{7es}\right)\left(\frac{5^2e^6 mr}{7^2 s^2}\right)^r\\
&\leq & \frac{1}{3} \quad \mbox{(if $s \geq  500 \sqrt{mn \lg \frac {2m}{n}}$)}.
\end{eqnarray*}

\newcommand{\event}{{\cal E}}

\item[(P2) holds.]  For (P2) to fail, there must be disjoint sets $S,
  R \subseteq U,$ where $|S|=n' \leq n$, $|R|=r\geq \ceil{n\lg \frac {2m}{n}}$
  for which the condition specified in Definition~\ref{def:admissible}
  does not hold. Then, $R=R_1 \cup R_2$, where $R_1=\{y \in R:
  |\Gamma_{G_3}(y) \cap \Gamma_{G_3}(S)|=1\}$ and $R_2=\{y \in R:
  |\Gamma_{G_3}(y) \cap \Gamma_{G_3}(S)| \geq 2\}$; let $r_1 = |R_1|$
  and $r_2=|R_2|$.  Furthermore, for $y \in R_1$ we have
  $|\Gamma_{G_1\cup G_2}(y) \cap \Gamma_{G_1 \cup G_2}(R \cup S \setminus\{y\})| \geq
  2$. This implies that $|\Gamma_{G_1\cup G_2}(R \cup S)| \leq
  3(n'+r)-r_1$. Fix $R_1 \subseteq R$, $R_2=R\setminus R_1$ and define
  events $\event_1$, $\event_2$ and $\event_\star$ as follows.
\begin{eqnarray*}
\event_1  & \equiv & \forall y \in R_1: |\Gamma_{G_3}(y) \cap \Gamma_{G_3}(S)|=1;\\
\event_2  & \equiv & \forall y \in R_2: |\Gamma_{G_3}(y) \cap \Gamma_{G_3}(S)|\geq2;\\
\event_\star  & \equiv & |\Gamma_{G_1\cup G_2}(R \cup S)| \leq  3(n'+r)-r_1.
\end{eqnarray*}
By the union bound, the probability that (P2) fails is bounded by the
sum of $\Pr[\event_1]\Pr[\event_2]\Pr[\event_\star]$ taken over all
valid choices of $S$, $R$, $R_1$ and $R_2$. We have
$\Pr[\event_1] \leq  \left(\frac{4n}{s}\right)^{r_1} \quad \mbox{and} \quad
\Pr[\event_2] \leq \left[{4 \choose 2} \left(\frac{n}{s}\right)^2\right]^{r_2} \leq  \left(\frac{3n}{s}\right)^{2r_2}.$ 
To bound $\Pr[\event_\star]$, we proceed as we did above for (P1). We have
\begin{eqnarray*}
\Pr[\event_\star] &\leq&  {3s \choose {3(n'+r)-r_1}} 
\left(\frac{\ell \ell_0 \ell_1}{s^3}\right)^{|S \cup R|} \\
&\leq& \left(\frac{3es}{3(n'+r) - r_1}\right)^{3(n'+r)-r_1} 
      \left(\frac{3(n'+r) - r_1}{3s}\right)^{3(n'+r)} \\
&\leq& \exp(3(n'+r) - r_1) 
 \left(\frac{3(n'+r)-r_1}{3s}\right)^{r_1}.
\end{eqnarray*}
Thus,
\begin{eqnarray*}
&& \Pr[\event_1]\Pr[\event_2]\Pr[\event_\star]\\
&\leq & 
   \left(\frac{4n}{s}\right)^{r_1}  
   \left(\frac{3n}{s}\right)^{2r_2} 
   \exp(3(n'+r) - r_1)
   \left(\frac{3(n'+r)-r_1}{3s}\right)^{r_1}\\
%&\leq& \left[ 9 \exp(3(n+r)/r) \left(\frac{3 n^{r_1+2r_2} (n+r)^{\frac{r_1}{r}}}{s^2}
%\right) \right]^r.
   &\leq& \left(9e^6\right)^r\left(\frac {n^{r+r_2} \left(n+r\right)^{r_1}} {s^{2r}}\right).
\end{eqnarray*}
Using the union bound, we conclude that
\begin{eqnarray*}
&&\Pr[\mbox{(P2) fails}]  \\
&\leq& \sum_{n'=1}^n \sum_{r \geq \ceil{n \lg \frac {2m}{n}}} {m \choose n'}{m \choose r} 
\left(9e^6\right)^r 
  \frac {n^{r}} {s^{2r}}\sum_{r_1=0}^r {r \choose r_1} n^{r_2} {\left(n+r\right)^{r_1}}\\ 
& \leq & \sum_{n'=1}^n \sum_{r \geq \ceil{n \lg \frac {2m}{n}}} \left(\frac {em} {n'}\right)^{n'} \left(\frac {em} {r}\right)^{r} 
\left(9e^6\right)^r 
  \frac {n^{r}} {s^{2r}} \left(2n+r\right)^r\\
&\leq& \sum_{n'=1}^n \sum_{r \geq \ceil{n \lg \frac {2m}{n}}}  
\left[\frac{m^{1+ \frac {n'} {r}} n^{1- \frac{n'}{r}} \left(9e^8\right)\left(2n+r\right)}
{rs^2}\right]^r\\
&\leq&\sum_{n'=1}^n \sum_{r \geq \ceil{n \lg \frac {2m}{n}}}  \left[\frac{3^3e^82mn}{s^2}\right]^r\\
&\leq& \frac{1}{3} \quad \mbox{(for $s \geq \ceil{500\sqrt{mn\lg \frac {2m}{n}}}$ and $m$ large)}.
\end{eqnarray*}
Thus, with probability at least $\frac{1}{3}$ the random graph $G$ is 
admissible.
\end{description}
\end{proof}

\section{Lower bound: Proof of Theorem~\ref{thm:multi-probe-lb}}
\label{sec:general-lb}

Our theorem follows immediately from the following lemma.
\begin{lemma}
Suppose $t \geq 2$, and $s$, $m$ and $n$ are such that (i) $4^t \leq n
\leq m^{\frac{1}{2(t-1)}}$ and (ii) $s \leq \frac{1}{15}
m^{\frac{1}{t-1}(1-\frac{4^{t}}{n})}$.  If there is a $t$-probe scheme
on a universe of size $m$ that uses space at most $s$, then there are
disjoint sets $S$ and $T$ of size at most $n$ each such that for every
assignment to the memory, some query in $S$ is answered with a `No' or
some query in $T$ is answered with a `Yes'.
\end{lemma}
\begin{proof}
For $t=2$, the claim is established in the proof of
Theorem~\ref{thm:two-probe} (see the last line of the proof). We will
use induction on $t$ to generalize this claim to larger values of
$t$. Assume the claim is true for $t=k-1$ and we wish to show that it
holds for $t=k$. Fix $m$, $n$ and a $k$-probe scheme that satisfy our
assumptions.  We now show how the sets $S$ and $T$ are obtained.
There is a cell to which at least $\frac{m}{s}$ of the elements make
their first probe: call this set of elements $U'$. By fixing the value
of this cell at $0$, we obtain a $(k-1)$-probe scheme for the universe
$U'$. We will verify that the assumptions needed for induction are
satisfied for this scheme. We conclude by induction that there are
disjoint sets $S_0, T_0 \subseteq U'$ each of size at most
$\frac{n}{2}$. Let $U''= U' \setminus (S_0 \cup T_0)$. Now, assume
that the cell has value $1$, and apply induction to the resulting
$(k-1)$-probe scheme (for the universe $U''$) to obtain sets $S_1$ and
$T_1$ of size at most $\frac{n}{2}$. Our claim for $t=k$ then follows
immediately by taking $S=S_0 \cup S_1$ and $T = T_0 \cup T_1$.

It remains to verify that the assumptions (i) and (ii) needed for the
induction hypothesis in fact do hold. Since $|U'| \geq |U''|$, it is
enough to verify the conditions for $U''$.  Now $|U'| \geq \frac{m}{s}
\geq 15 m^{\frac{k-2}{k-1}} \geq 15 n^{2(k-2)} \geq 2n$.  Thus, $m'=
|U''|\geq |U'| -n \geq \frac{m}{2s}$.  We need to find sets of size at
most $n'=\floor{\frac{n}{2}}$.  Clearly $n' \geq \frac{n}{4} =
4^{k-1}$, so condition (i) holds. Also, since $m' \geq \frac{m}{2s}$
and $n' \geq \frac{n}{4}$, we have $s \leq
\frac{1}{15}m'^{\frac{1}{k-2}(1-\frac{4^{k-1}}{n'})}$, so condition
(ii) holds.
\end{proof}

\section{General upper bound: non-adaptive (Theorem~\ref{thm:generalnonadaptive})}
\label{sec:multi-probe-ub}

\begin{definition}
A non-adaptive $(m,s,t)$-graph is a bipartite graph $G$ with vertex
sets $U=[m]$ and $V$ ($|V|=ts$). $V$ is partitioned into $t$ disjoint
sets: $V_1,\ldots,V_t$; each $V_i$ has $s$ vertices. Every 
$u\in U$ has a unique neighbour in each $V_i$. A non-adaptive
$(m,s,t)$-graph naturally gives rise to a non-adaptive $(m, ts, t)$-query
scheme ${\cal T}_G$ as follows. We view the memory (an array $L$ 
of $ts$ bits) to be indexed by vertices in $V$. On receiving the query
``Is $u$ in $S$?'', we answer ``Yes'' iff the Majority of the locations in the 
neighbourhood of $u$ contain a $1$. We say that the query scheme ${\cal T}_G$ 
is satisfiable for a set  $S\subseteq[m]$, if there is an assignment to the memory locations 
$(L[v] : v \in V)$, such that ${\cal T}_G$ correctly answers all 
queries of the form ``Is $x$ in $S$?''.
\end{definition}

We now restrict attention to odd $t\geq5$. First, we identify an
appropriate property of the underlying non-adaptive $(m,s,t)$-graph
$G$ that guarantees that ${\cal T}_G$ is satisfiable for all sets $S$
of size at most $n$.  We then show that such a graph exists for some
$s=O(m^{\frac{2}{t-1}} n^{1-\frac{2}{t-1}}\lg \frac {2m}{n})$.

\begin{definition}[Non-adaptive admissible graph] \label{def:t-admissible}
We say that a non-adaptive $(m,s,t)$-graph $G$ is admissible for sets of
size at most $n$ if the following two properties hold:
\begin{enumerate}
	\item[(P1)] $\forall R \subseteq [m]$ ($|R|\leq n + \ceil{2n \lg \frac {2m}{n}}$):
                $|\Gamma_{G}(R)| \geq \frac {t+1} {2} |R|$, where $\Gamma_G(R)$
                is the set of neighbors of $R$ in $G$.
        \item[(P2)] $\forall S \subseteq [m]$ ($|S| = n$): $|T_S| \leq \ceil{2n
                \lg \frac {2m}{n}}$, where $T_S = \{ y \in [m] \setminus S : |\Gamma_G(y)
                        \cap \Gamma_G(S)| \geq \frac {t+1}{2} \}$.
\end{enumerate}
\end{definition}
Our theorem will follow from the following claims.
\begin{lemma} \label{lm:t-admissibleImpliesSatisfiable}
If a non-adaptive $(m,s,t)$-graph $G$ is admissible for sets of size
at most $n$, then the non-adaptive $(m, ts, t)$-query scheme ${\cal
  T}_{G}$ is satisfiable for every set $S$ of size at most $n$.
\end{lemma}
\begin{lemma} \label{lm:t-admissibleExists}
There is a non-adaptive $(m,s,t)$-graph, with $s=O(m^{\frac{2}{t-1}}
n^{1-\frac{2}{t-1}}\lg \frac {2m}{n})$, that is admissible for every set
$S\subseteq [m]$ of size at most $n$. 
\end{lemma}

\begin{proof}[Proof of Lemma~\ref{lm:t-admissibleImpliesSatisfiable}]
Fix an admissible graph $G$. %as promised by Lemma~\ref{lm:t-admissibleExists}.  
Thus, $G$ satisfies (P1) and (P2)
above. Fix a set $S \subseteq [m]$ of size at most $n$.  We will show
that there is a 0-1 assignment to the memory such that all queries are
answered correctly by ${\cal T}_G$.

Let $S' \subseteq [m]$ be such that $S \subseteq S'$ and $|S'| =
n$. From (P2), we know $|T_{S'}| \leq \ceil{2n \lg \frac {2m}{n}}$. Hence, $|S'
\cup T_{S'}| \leq n + \ceil{2n \lg \frac {2m}{n}}$. From (P1) and Hall's theorem,
we may assign to each element $u \in S' \cup T_{S'}$ a set $A_u
\subseteq V$ such that (i) $|A_u| = \frac {t+1}{2}$ and (ii) the
$A_u$'s are disjoint.  For each $u\in S \subseteq S'$, we assign the
value 1 to all locations in $A_u$.  For each $u\in (S' \cup T_{S'})
\setminus S$, we assign the value 0 to all locations in $A_u$. Since
$\frac {t+1}{2} > \frac t 2$, all queries for $u\in S' \cup T_{S'}$ are
answered correctly.

Assign 0 to all locations in $\Gamma_G([m]
\setminus (S' \cup T_{S'}))$. For $y\in [m]
\setminus (S' \cup T_{S'})$, $|\Gamma_G(y) \cap \Gamma_G(S)|\leq \frac {t-1} {2}$.
As a result, queries for elements in $[m] \setminus 
(S' \cup T_{S'})$ are answered correctly, as the majority evaluates to 0 for each 
one of them. 
\end{proof}

\begin{proof}[Proof of lemma~\ref{lm:t-admissibleExists}]
In the following, set 
\[ s = \ceil{60 m^{\frac{2}{t-1}} n^{1-\frac{2}{t-1}}\lg \frac {2m}{n}}.\]
We show that a suitable random non-adaptive $(m,s,t)$-graph $G$ is
admissible for sets of size at most $n$ with positive probability. The
graph $G$ is constructed as follows. Recall that $V = \bigcup_i
V_i$. For each $u \in U$, one neighbor is chosen uniformly and
independently in each $V_i$.
\begin{description}
\item[(P1) holds.] If (P1) fails, then for some non-empty $W \subseteq
  U$, $(|W| \leq n + \ceil{2n \lg \frac {2m}{n}})$, we have $|\Gamma_G(W)| \leq
  \frac {t+1}{2}|W| - 1$. Fix a set $W$ of size $r \geq 1$ and $L
  \subseteq V$ of size $\frac{t+1}{2}r - 1$. Let $L$ have
  $\ell_i$ elements in $V_i$; thus, $\sum_i \ell_i =
  \frac{t+1}{2}r - 1$. Then,
\[
\Pr[\Gamma_G(W) \subseteq L] \leq \prod_{i=1}^t
\left(\frac{\ell_i}{|V_i|}\right)^{r} \leq
\left(\frac{(\frac{t+1}{2})r-1}{ts}\right)^{tr},
\]
where the last inequality is a consequence of GM $\leq$ AM. We conclude, using the union bound over choices of $W$ and $L$, that (P1) fails with probability 
at most
\begin{align}
& \sum_{r=1}^{n+\ceil{2n\lg \frac {2m}{n}}} {m\choose r} 
{{ts} \choose \frac{t+1}{2}r-1}\left(\frac{\frac{t+1}{2}r-1}{ts}\right)^{tr}\\
&\leq 
  \sum_{r=1}^{n+\ceil{2n\lg \frac {2m}{n}}}
  \left(\frac{em}{r}\right)^r\left(\frac{tes}{\frac{t+1}{2}r-1}\right)^
  {\frac{t+1}{2}r-1}  
 \left(\frac{\frac{t+1}{2}r-1}{ts}\right)^{tr} \nonumber \\
&\leq
\sum_{r=1}^{n+\ceil{2n\lg\frac {2m}{n}}}
\left[\frac{(e^{\frac{t+3}2 - \frac 1 r}) m  r^{\frac{t-1}2-1 + \frac{1}{r}} }{ 
(s^{\frac 1 r})s^{\frac{t-1}2}} \right]^r \leq \frac{1}{3}, %\ref{eq:non-adaptivebound1}%
\end{align}
where the last inequality holds because we have chosen $s$ large enough.

\item[(P2) holds.] For (P2) to fail, there must exist a set $S\subseteq [m]$ 
of size $n$ such that $|T_S| > \ceil{2n \lg \frac {2m}{n}}$. Fix a set $S$ of size $n$.
Fix a $y \in [m] \setminus S$.
\[
\Pr[y \in T_S]\leq {t \choose \frac {t+1}2}\left(\frac n s\right)^{\frac {t+1}2}
\leq \frac {n} {10m},
\]
where the last inequality holds because of choice of $s$ and $m$ is
large. Thus, $\E[|T_S|]  \leq \frac {n}{10}$.
To conclude that $|T_S|$ is bounded with high probability, we will use the following version of Chernoff bound: if $X = \sum_{i=1}^N X_i$, where each random variable $X_i \in \{0,1\}$ independently, then if $\gamma > 2e \E[X]$, then  $\Pr[X > \gamma] \leq 2^{-\gamma}$. Then,  for all large $m$, 
\[ \Pr[ |T_S| > 2n \lg \frac {2m}{n}] \leq 2^{- 2n \lg \frac {2m}{n}}. \]
Using the union bound, we conclude that
\begin{eqnarray*}
&&\Pr[\mbox{(P2) fails}]\\
&\leq& \left( \frac {em} n \right)^n 2^{-2n \lg \frac {2m}{n}}\\
&\leq& \frac 1 3.
\end{eqnarray*}
Thus, with probability at least $\frac 1 3$ the random graph $G$ is admissible.
\end{description}
\end{proof}

%%%%%%%%%%%%%%%%%%%%%%%%%%%%%%%%%%%%%

\section{General upper bound: adaptive (Theorem~\ref{thm:generaladaptive})}
\label{sec:multi-probe-ub-adaptive}
\newcommand{\leaves}{\mathsf{leaves}}

In order to show that $s(m,n,t)$ is small, we will exhibit efficient
adaptive schemes to store sets of size {\em exactly}
$n$. This will
imply our bound (where we allow sets of size {\em at most} $n$)
because we may {\em pad} the universe with $n$ additional elements,
and extend $S$ ($|S| \leq n$)by adding $n-|S|$ additional elements, to
get a subset is of size exactly $n$ in a universe of size $m+n \leq
2m$.
\begin{definition} 
An adaptive $(m,s,t)$-graph is a bipartite graph $G$ with vertex sets $U=[m]$
and $V$ ($|V|=(2^t-1)s$). $V$ is partitioned into $2^t-1$ disjoint sets:
$A$, $A_0$, $A_1$, $A_{00}$,\ldots, that is, one $A_{\sigma}$ for each $\sigma
\in \{0,1\}^{\leq (t-1)}$; each $A_{\sigma}$ has $s$ vertices. Between
each $u \in U$ and each $A_{\sigma}$ there is exactly one edge. 
%For
%$i=1,\ldots,t$, the subgraph of $G$ induced by $U$ and
%$V_i=\cup_{\sigma: |\sigma|=i-1} A_\sigma$ will be referred to as
%$G_i$.
Let $V_i:=\cup_{\sigma: |\sigma|=i-1} A_\sigma$.
An $(m,s,t)$-graph naturally gives rise to a systematic
$(m,(2^t-1)s, t)$-query scheme ${\cal T}_G$ as follows. We view the
memory (an array $L$ of $(2^t-1)s$ bits) to be indexed by vertices in
$V$. For query element $u\in U$, if the first ${i-1}$ probes resulted
in values $\sigma \in \{0,1\}^{i-1}$, then the $i$-th probe is made to
the location indexed by the unique neighbor of $u$ in $A_{\sigma}$. In
particular, the $i$-th probe is made at a location in $V_i$. 
We answer ``Yes'' iff the last bit read is $1$.
We refer to $V_t$ as the leaves of $G$ and for $y \in [m]$, let  
$\leaves(y):=V_t \cap \Gamma_G(y)$. For $R \subseteq [m]$, let
$\leaves(R):=V_t \cap \Gamma_G(R)$.

We say that the query scheme ${\cal T}_{G}$ is satisfiable for a set
$S \subseteq [m]$, if there is an assignment to the memory locations
$(L[v]: v \in V)$, such that ${\cal T}_{G}$ correctly answers all
queries of the form ``Is $x$ in $S$?''.
\end{definition}

We assume that $t \geq 3$ is odd and show that $\forall \epsilon >0$ 
$\forall n \leq m^{1-\epsilon}$ $\forall t \leq \frac{1}{10}\lg\lg m$
$s(m,n,t)=O(\exp(e^{2t})m^{\frac 2 {t+1}} n^{1 - \frac 2 {t+1}} \lg m)$. 
Our $t$-probe scheme will
have two parts: a $t_1$-probe non-adaptive part and a $t_2$-probe
adaptive part, such that $t_1 + t_2 = t$. The respective parts will be
based on appropriate non-adaptive $(m,s,t_1)$-graph $G_1$ and adaptive
$(m,s,t_2)$-graph $G_2$ respectively. To decide set membership, we
check set membership in the two parts separately and take the AND,
that is, we answer ``Yes'' iff all bits read in ${\cal T}_{G_1}$ are $1$
and the last bit read in ${\cal T}_{G_2}$ is $1$.  We refer to this
scheme as ${\cal T}_{G_1} \wedge {\cal T}_{G_2}$.

First, we identify appropriate properties of the underlying graphs
$G_1$ and $G_2$ that guarantee that all queries are answered correctly
for sets of size $n$.  We then show that such graphs exist with
$s=O(\exp(e^{2t}-t)m^{\frac 2 {t+1}} n^{1 - \frac 2 {t+1}} \lg m)$.

\newcommand{\surv}{\mathsf{survivors}} We will use the following
constants in our calculations: $\alpha:= 2^{t_2}-1$ and $\beta:=
2^{t_2} - t_2$.  Note that $\alpha$ is the total number of nodes in a
$t_2$-probe adaptive decision tree. In any such decision tree, for
every choice of $\beta$ nodes and every choice $b \in \{0,1\}$ of the
answer, it is possible to assign values to those $\beta$ nodes so that
the decision tree returns the answer $b$.

\begin{definition}[admissible-pair] \label{def:admissible-pair}
We say that a non-adaptive $(m,s,t_1)$-graph $G_1$ and an adaptive
$(m,s,t_2)$-graph $G_2$ form an admissible pair $(G_1,G_2)$ for
sets of size $n$ if the following conditions hold.
\begin{enumerate}
\item[(P1)]
 $\forall S \subseteq [m]$ ($|S| = n$):
$|{\surv}(S)| \leq 10 m \left( \frac n s \right)^{t_1}$,
where ${\surv}(S) = \{y \notin S: \Gamma_{G_1}({y})\subseteq \Gamma_{G_1}(S)\}$.
\item[(P2)] For $S \subseteq [m]$ ($|S| = n$), let 
${\surv}^+(S) = \{y \in {\surv}(S): 
{\leaves_{G_2}}(S) \cap {\leaves_{G_2}}(y) \neq \emptyset \}$. 
%{\Gamma_{G_2}}_{t_2}(S) \cap {\Gamma_{G_2}}_{t_2}({y}) \neq \emptyset \}$. 
%Then, 
Then, 
$\forall S \subseteq [m]$ ($|S| = n$) 
$\forall T \subseteq S \cup {\surv}^+(S)$:
$\Gamma_{G_2}(T) \geq \beta |T|$.
\end{enumerate}
\end{definition}

%%%%%%%%%%%%%%%%%%%%%%%%%%%%%%%%%%%

\begin{lemma} \label{lm:t1t2-admissibleImpliesSatisfiable}
If a non-adaptive $(m,s,t_1)$-graph $G_1$ and an adaptive $(m,s,t_2)$-graph 
$G_2$ form an admissible pair for sets of size $n$, then 
the query scheme ${\cal T}_{G_1} \wedge {\cal T}_{G_2}$ is satisfiable 
for every set $S\subseteq [m]$ of size $n$.
\end{lemma}

\begin{lemma} \label{lm:t1t2-admissibleExists}
Let $t \geq 3$ be an odd number; let $t_1 = \frac{t-3}{2}$ and $t_2=\frac{t+3}{2}$.  
Then, there exist an admissible pair of graphs consisting of a non-adaptive $(m,s,t_1)$-graph $G_1$ and an adaptive $(m,s,t_2)$-graph $G_2$ with $s=O(\exp(e^{2t}-t)m^{\frac 2 {t+1}} n^{1 - \frac 2 {t+1}} \lg m)$.
\end{lemma}

\begin{proof}[Proof of Lemma~\ref{lm:t1t2-admissibleImpliesSatisfiable}]
Fix an admissible pair $(G_1, G_2)$. 
Thus, $G_1$ satisfies (P1) and $G_2$ satisfies (P2) above. Fix a set
$S \subseteq [m]$ of size $n$. We will show that there is an assignment such that
${\cal T}_{G_1} \wedge {\cal T}_{G_2}$ answers all questions of the form ``Is $x$ in $S$?'' correctly.

The assignment is constructed as follows. Assign 1 to all locations in $\Gamma_{G_1}(S)$ and 0 to the remaining locations in $\Gamma_{G_1}(S)$. Thus, ${\cal T}_{G_1}$ answers ``Yes'' for all query elements in $S$ and answers ``No'' for all query elements outside $S \cup \surv(S)$. However, it (incorrectly) answers ``Yes'' for elements in $\surv(S)$. We will now argue that these {\em false positives} can be eliminated using the scheme ${\cal T}_{G_2}$.

Using (P2) and Hall's theorem, we may assign to each element $u \in S
\cup {\surv}^+(S)$ a set $L_u \subseteq V(G_2)$ such that (i)$|L_u| =
\beta$ and (ii) the $L_u$'s are disjoint. Set $b_u = 1$ for $u\in S$ and
$b_u = 0$ for $u \in \surv^+(S)$ (some of the false positives).  As observed
above for each $u \in S \cup \surv^+(S)$ we may set the values in the
locations in $L_u$ such that the value returned on the query element
$u$ is precisely $b_u$.  Since the $L_u$'s are disjoint we may take such
an action independently for each $u$.  After this partial assignment,
it remains to ensure that queries for elements $y \in
\surv(S)\setminus \surv^+(S)$ (the remaining false positives) return a ``No''. Consider any such 
$y$. By the definition of $\surv^+(S)$, no location in
$\leaves_{G_2}(y)$ has been assigned a value in the above partial
assignment. Now, assign 0 to all unassigned locations in $V(G_2)$. Thus
${\cal T}_{G_2}$ returns the answer ``No'' for queries from
${\surv}(S) \setminus {\surv}^+(S)$.
\end{proof}
%%%%%%%%%%%%%%%%%%%%%%%%%

\begin{proof}[Proof of Lemma~\ref{lm:t1t2-admissibleExists}] In the following, let 
\[ s = \ceil{\exp(e^{2t}-t) m^{\frac{2}{t+1}} n^{1 - \frac{2}{t+1}}\lg m}.\]
We will construct the non-adaptive $(m,s,t_1)$-graph $G_1$ and the 
$(m,s,t_2)$-graph $G_2$ randomly, and show that with positive probability the pair $(G_1,G_2)$ is admissible. 
The graph $G_1$ is constructed as in the proof of 
Lemma~\ref{lm:t-admissibleExists}, and the analysis is similar.  Recall that $V(G_1) = \bigcup_{i \in [t_1]} V_i(G_1)$.
For each $u \in U$, one neighbor is chosen uniformly and independently from each $V_i(G_1)$.

\paragraph{(P1) holds.} Fix a set $S$ of size $n$. Then, 
$\E[|{\surv}(S)|] \leq (m-n)\left(\frac n s\right)^{t_1} 
\leq m \left( \frac n s \right)^{t_1}$.
As before, using the Chernoff bound, we conclude that
\[ \Pr[|{\surv}(S)| > 10 m \left( \frac n s \right)^{t_1}] \leq 
2^{-10m\left(\frac n s  \right)^{t_1}}.\]
Then, by the union bound,
\begin{eqnarray*}
	\Pr[\mbox{P1 fails}]&\leq& {m \choose n}2^{-10m\left(\frac n s  \right)^{t_1}} \\
&\leq& \frac 1 {10},
\end{eqnarray*}
where the last inequality follows from our choice of $s$.

Fix a graph $G_1$ such that (P1) holds. The random graph $G_2$ is
constructed as follows.  Recall that $V({G_2})= \bigcup_{z \in
  \{0,1\}^{\leq t_2 - 1}} A_z$.  For each $u \in [m]$, one neighbor is
chosen uniformly and independently from each $A_z$.

To establish (P2), we need to show that all sets of the form $S' \cup
R$, where $S' \subseteq S$ and $R \subseteq \surv^+(S)$ expand. To
restrict the choices for $R$, we first show in
Claim~\ref{cl:survivorplus} (a) that with high probability
$\surv^+(S)$ is small. Then, using direct calculations, we show that
whp the required expansion is available in the random graph $G_2$.
\newcommand{\event}{\mathcal{E}}
\begin{claim} \label{cl:survivorplus} 
\begin{enumerate}
\item[(a)] Let $\event_a \equiv \forall S \subseteq [m] (|S| = n):
  |{\surv}^+(S)| \leq 100 \cdot 2^{t_2} m\left(\frac n s \right)^{t_1 + 1}$;
  then, $\Pr[\event_a] \geq \frac{9}{10}$.
\item[(b)] Let $\event_b \equiv \forall R \subseteq [m]\, (|R| \leq n
  + \ceil{n \lg m}): |\Gamma_{G_2}(R)| \geq \beta|R|$; then,
  $\Pr[\event_b] \geq \frac{9}{10}$.
\item[(c)] Let $\event_c = \forall S \subseteq [m] (|S| = n), \forall
  S' \subseteq S, \forall R \subseteq {\surv}^+(S) (\ceil{n \lg m} \leq
  |R| \leq 100 \cdot 2^{t_2} m\left(\frac n s \right)^{t_1 + 1}):
  |\Gamma_{G_2}(S' \cup R)| \geq \beta |S' \cup R|$; then,
  $\Pr[\event_c] \geq \frac{9}{10}$.
\end{enumerate}
\end{claim}
\paragraph{Proof of claim ~\ref{cl:survivorplus}.}  Part (a) follows by a routine application of Chernoff bound, as in several previous proofs. For a set $S$ of size $n$, we have $\E[{\surv}^+(S)]\leq |\surv(S)|2^{t_2}(\frac n s)
        \leq 2^{t_2}10m\left( \frac n s \right)^{t_1 + 1}$.  Then,
\begin{eqnarray*}
\Pr[\neg \event_a]&\leq& {m \choose n}2^{-2^{t_2}10m\left(\frac n s
\right)^{t_1+1}} \\ 
%&\leq& \exp(n \lg m -2^{t_2}10m\left(\frac n s
%\right)^{t_1+1})\\ 
&\leq& \frac 1 {10},
\end{eqnarray*}
where the last inequality holds because of our choice of $s$.

Next consider part (b).  If $\event_b$ does not hold, then for some
non-empty $W \subseteq [m]$, $(|W| \leq n + \ceil{n \lg m})$, we have
$|\Gamma_{G_2}(W)| \leq \beta |W| -1$. Fix a set $W$ of size $r \geq
1$ and $L \subseteq V(G_2)$ of size $\beta r-1$. Let $L$ have
$\ell_z$ elements in $A_z$. Then,
\[
\Pr[\Gamma_{G_2}(W) \subseteq L] \leq \prod_z \left(\frac{\ell_z}{|A_z|}\right)^{r} \leq
\left(\frac{\beta r-1}{\alpha s}\right)^{ \alpha r}.
\]
We conclude, using the union bound over choices of $W$ and $L$, that the probability that 
$\event_b$ does not hold is at most
\begin{eqnarray*}
&& \sum_{r=1}^{n+\ceil{n\lg m}} {m\choose r}
{{\alpha s} \choose \beta r-1}\left( \frac{\beta r-1}{\alpha s} \right)^{\alpha r}\\
&\leq&
\sum_{r=1}^{n+\ceil{n\lg m}}
\left(\frac {em}{r}\right)^r \left( \frac{\alpha es}{\beta r-1} \right)^{\beta r-1}
\left(\frac{\beta r-1}{\alpha s}\right)^{\alpha r}\\
&\leq&
\sum_{r=1}^{n+\ceil{n\lg m}}
\left(\frac{\beta r}{\alpha es}\right) \left[ \frac {em}{r} e^\beta \left(\frac {\beta r}{\alpha s}\right)^{\alpha-\beta} \right]^r\\
&\leq&
\sum_{r=1}^{n+\ceil{n\lg m}}
\left(\frac{\beta r}{\alpha es}\right) \left[ e^{\beta +1}\left( \frac {\beta}{\alpha} \right)^{\alpha-\beta}  \left(\frac {mr^{\alpha-\beta-1}}{s^{\alpha-\beta}}\right) \right]^r\\
&\leq& \frac {1}{10},
\end{eqnarray*}
where the last inequality holds because of our choice of $s$.

Finally, we justify part (c).  To bound the probability that
$\event_c$ fails, we consider a set $S \subseteq [m]$ of size $n$, a
subset $S' \subseteq S$ of size $i$ (say), a subset $R \subseteq {\surv}^+(S)$ of
size $r$ (where $\ceil{n\lg m} \leq r \leq 100 \cdot 2^{t_2} m\left(\frac n
s \right)^{t_1 + 1}$) and $L \subseteq V(G_2)$ of size $ \ell = 
\beta (i+r) $ and define the event
\[ \event(S,S',R,L) \equiv (\forall y \in R: \leaves_{G_2}(S) \cap \leaves_{G_2}(y) \neq \emptyset) \wedge \Gamma_{G_2} (S' \cup R) \subseteq L.\]
Then,
\begin{align}
 \Pr[\event(S,S',R,L)] &\leq  \left( \frac {2^{t_2}n}{s} \right)^r 
\left( \frac {\ell}{(\alpha-1) s} \right)^{(\alpha-1)r} 
\left( \frac {\ell}{\alpha s} \right)^{\alpha i}\\
&\leq  \left( \frac {2^{t_2}n}{s} \right)^r
      \left( \frac {\beta(i+r)}{(\alpha-1) s} \right)^{(\alpha-1)(i+r)} 
     \left( \frac{\beta(i+r)}{\alpha s} \right)^{i},
\end{align} 
where the factor $\left( \frac {2^{t_2}n}{s} \right)^r $ is justified
because of the requirement that every $y \in R$ has at least one
neighbour in $\leaves_{G_2}(S)$; the factor $\left( \frac
{\ell}{(\alpha-1) s} \right)^{(\alpha-1)r}$ is justified because all
the remaining neighbours must lie in $L$ (we use AM $\geq$ GM); the
last factor $\left( \frac {\ell}{\alpha s} \right)^{\alpha i}$ is
justified because all neighbors of elements in $S’$ lie in $L$ (again
we use AM $\geq$ GM). To complete the argument we apply the union
bound over the choices of $(S,S',R,L)$. Note that we may restrict
attention to $\ell = \beta (i+r)$ (because for our choice of $s$, we
have $\beta(i+r) \leq |V(G_2)| = \alpha s$).  Thus, the probability
that $\event_c$ fails to hold is at most
     \[ \sum_{S,S',R,L} \Pr[\event_c(S,S',R,L)],\] 
     where $S$ ranges over sets of size $n$, $S' \subseteq S$ of size $i$, $R
     \subseteq \surv(S)$ of size $r$ such that $\floor{n \lg m} \leq
     r \leq 100 2^{t_2}m\left( \frac {n}{s} \right)^{t_1 + 1}$, $L$ is
     a subset of $V(G_2)$ of size $\beta(i+r)$. We evaluate this sum
     as follows.
\begin{align}
 & \sum_{r} \sum_{i} {m \choose n}
        {\floor{ 10m \left( \frac {n}{s} \right)^{t_1}} \choose r}
        {n \choose i}
        {\alpha s \choose \beta (i+r)}
        \left( \frac {2^{t_2}n}{s} \right)^{r}
        \left(\frac {\beta(i+r)}{(\alpha-1) s} \right)^{(\alpha-1)(i+r)} 
        \left(\frac{\beta(i+r)}{\alpha s} \right)^{i}
        \\
%%%%%%%%%%%%%%%%%%%%%%%%%%%%%%%%%%%%%%%%%%%%%%%%
&\leq
  \sum_{r} \sum_{i} 
   \left[ \left(\frac {em}{n} \right)^{\frac {n}{i+r}} 
                \left( \frac {10 e m \left( \frac {n}{s} \right)^{t_1}} {r} \right)^{\frac {r}{i+r}}
                %\left( \frac {en}{i} \right)^{\frac {i}{i+r}}  
		{n \choose i}^{\frac {1}{1+r}}
                \left( \frac {\beta(i+r)}{(\alpha-1)s} \right)^{\alpha-1} \right.
\nonumber\\
&
\left.
        \left(\frac{e\alpha s}{\beta(i+r)}\right)^\beta
        \left( \frac {2^{t_2}n}{s} \right)^{\frac{r}{i+r}}
        \left(\frac {\beta(i+r)}{(\alpha-1) s} \right)^{\frac{i}{i+r}} 
\right]^{i+r}\\
&\leq
  \sum_{r} \sum_{i} 
   \left[ \left(\frac {em}{n} \right)^{\frac {n}{i+r}} 
                \left( \frac {10 e m \left( \frac {n}{s} \right)^{t_1}} {r} \right)^{\frac {r}{i+r}}
                %\left( \frac {en}{i} \right)^{\frac {i}{i+r}}  
                {n \choose i}^{\frac {1}{1+r}}
		\left( \frac {\beta(i+r)}{(\alpha-1)s} \right)^{\alpha-\beta-1} \right.
\nonumber\\
&
\left.
                \left(\frac{e\alpha}{\alpha-1} \right)^\beta
                \left( \frac {\beta(i+r)}{2^{t_2}n (\alpha-1)} \right)^{ \frac {i}{i+r} } 
                \left( \frac {2^{t_2}n}{s} \right)   
\right]^{i+r}.
\end{align}
We will show that the quantity inside the square brackets is at most
$\frac{1}{2}$. Then, since $r \geq n\lg m$ and $i \geq 0$
\[ \Pr[\neg \event_c] \leq \left(\sum_r 2^{-r}\right) \left(\sum_i 2^{-i} \right) \leq \frac{1}{10}.\]
The quantity in the brackets can be decomposed as a product of two
factors, which we will bound separately.
\begin{description}
\item[Factor 1:] Consider the following contributions
\[\left(\frac {em}{n} \right)^{\frac {n}{i+r}} 
  (10 e)^{\frac{r}{i+r}} 
  %\left(\frac {en}{i} \right)^{\frac {i}{i+r}}
  {n \choose i}^{\frac {1}{i+r}}
  \left(\frac{e\alpha}{\alpha-1}\right)^{\beta}
  \left(\frac {\beta(i+r)}{2^{t_2}n(\alpha-1)} \right)^{\frac{i}{i+r}}.
\]
Since $r \geq n \lg m$ and $i \leq n$, we have $\frac{i}{i+r} \leq
\frac{n}{n+r} \leq \frac{1}{\lg m} \leq \frac{1}{\lg_e m}$. Thus, for all large enough $m$, 
this quantity is at most
\[ e^2 \cdot 10e \cdot e^2\cdot (2e)^{\beta} \cdot e \leq \exp(e^{2t}-t).\]

\item[Factor 2:] We next bound the contribution for the remaining factors.
\begin{align}
& \left( \frac{m (\frac{n}{s})^{t_1}}{r}\right)^{\frac{r}{i+r}}
\left(\frac{\beta(i+r)}{(\alpha -1) s}\right)^{\alpha-\beta -1}\left(\frac{2^{t_2} n}{s}\right)\\
&\leq \left( \frac{m (\frac{n}{s})^{t_1}}{r}\right)
\left(\frac{2r}{ s}\right)^{\alpha-\beta -1}
\left(\frac{2^{t_2} n}{s}\right) \label{eq:dropexponent} \\
&= \frac{mn^{t_1+1}2^{\alpha-\beta+t_2 -1} r^{\alpha-\beta -2}}{s^{\alpha-\beta +t_1}}.
\end{align}
To justify (\ref{eq:dropexponent}),
Recall that $r \leq 100 \cdot 2^{t_2} m\left(\frac n s \right)^{t_1 + 1}$
and $s = \ceil{\exp(e^{2t}-t) m^{\frac{2}{t+1}} n^{1 - \frac{2}{t+1}}\lg m}$; thus $\frac{m (\frac{n}{s})^{t_1}}{r} \geq 1$.
Then, the above quantity is bounded by
\begin{align}
& \frac{mn^{t_1+1}2^{2(t_2-1)} \left(100 \cdot 2^{t_2} m n^{t_1+1}\right)^{\alpha-\beta-2}}{s^{(t_1+1)(\alpha-\beta - 2)} s^{\alpha - \beta + t_1}}\\
&\leq \left(\frac{100\cdot 2^{2t_2} m n^{t_1+1}}{s^{t_1+2}}\right)^{\alpha-\beta-1}.
\end{align}
\end{description}
Thus, since $s = \ceil{\exp(e^{2t}-t) m^{\frac{2}{t+1}} n^{1 -
    \frac{2}{t+1}}\lg m}$, then the product of the factors is at most
$\frac{1}{10}$, as required.
\end{proof}

\subsection*{Acknowledgments}

We are grateful to Pat Nicholson and Venkatesh Raman for their
comments on these results, and for sharing with us their recent
work~\cite{LMNR}.

%
%
%
%\begin{thebibliography}{1}
%\bibitem{AF2009} Noga Alon, Uriel Feige: On the power of two, three
%  and four probes. SODA 2009: 346--354.
%
%\bibitem{BMRV2002} Harry Buhrman, Peter Bro Miltersen, Jaikumar Radhakrishnan,
%  Srinivasan Venkatesh: Are Bitvectors Optimal? SIAM J. Comput. 31(6),
%  2002: 1723--1744.
%\end{thebibliography}

\end{document}